\documentclass[12pt]{amsart}
\usepackage{amsfonts,amssymb,amsthm,amsmath} %bbold}
\usepackage[latin1]{inputenc}
\usepackage{hyperref}
\usepackage{dsfont}

\newtheorem{lemma}{Lemma}

\newtheorem{theorem}[lemma]{Theorem}

\newtheorem{assumption}[lemma]{Assumption}
\newtheorem{remark}[lemma]{Remark}

\newcommand{\fH}{\mathfrak{H}}
\newcommand{\cH}{\mathcal{H}}
\newcommand{\cL}{\mathcal{L}}

\newcommand{\fD}{\mathfrak{D}}

\newcommand{\X}{\mathcal{X}}
\newcommand{\Y}{\mathcal{Y}}
\newcommand{\Z}{\mathcal{Z}}

\begin{document}

\date{4th May 2016}

\title[Ground state of the supermembrane and matrix models]{The ground state of the $D=11$ supermembrane and matrix models on compact regions}
\author[L Boulton]{Lyonell Boulton}
\author[M P Garcia del Moral]{Maria Pilar Garcia del Moral}
\author[A Restuccia]{Alvaro Restuccia}

\begin{abstract}
We establish a general framework for the analysis of boundary value problems of matrix models at zero energy on compact regions. We derive existence and uniqueness of ground state wavefunctions for the mass operator of the $D=11$ regularized supermembrane theory, that is the
$\mathcal{N}=16$  supersymmetric $SU(N)$ matrix model, on balls of finite radius. Our results rely on the structure of the associated Dirichlet form  and a factorization in terms of the supersymmetric charges. They also rely on the polynomial structure of the potential and various other supersymmetric properties of the system.
\end{abstract}

\maketitle

\clearpage

\tableofcontents
\section{Introduction} \label{intro}
Physical theories subject to boundary conditions play a crucial role in the study of physical properties at high energies and they are also notably relevant in the study of some condensed matter effects. Recently,  supersymmetric boundary conditions have received renewed attention, due to their connection with quantum phase transition at the boundary of topological superconductors. We refer to \cite{manolo1} and \cite{manolo2}, where a systematic analysis of a class of $N=2$ supersymmetric theories subject to several boundary conditions was carried out.

Matrix models related to the regularization of field theories with boundary conditions have been studied in order to test the AdS/CFT conjecture at finite temperature. In part this is due to its relation in the bulk picture with the black hole microstates in this regime \cite{lin-yin}. An analysis of the unbounded matrix model wavefunction, related to the (0+1) supersymmetric Yang-Mills ground state, has been considered in various other cases \cite{halpern}-\cite{wosiek1}. In the context of  M-theory and for the regularized supermembrane, this has been studied in \cite{dwhn} and subsequent work. For a matrix model wavefunction perspective see \cite{fh}-\cite{bhs}, for different aspects of Lorentz invariance see \cite{hlt}-\cite{wosiek2}, for inner solutions see \cite{hlt2} and for asymptotic solutions see \cite{halpern-su2}-\cite{fghhy}.

M-theory seen as a unification theory should provide  a quantum description of $D=11$ supergravity. In this setting, String Theory is regarded as a perturbative limit of M-theory which should include all non-perturbative effects. The $D=11$ supermembrane describes relevant degrees of freedom of M-theory, because it couples consistently to a $D=11$ supergravity background without destroying the local fermionic symmetry \cite{bst}. This provides strong evidence that the ground state of the $D=11$ supermembrane should correspond to a wavefunction constructed in terms of the $D=11$ supergravity multiplet.

In spite of several insightful attempts \cite{dwhn}-\cite{fghhy}, a construction of the ground state wavefunction of the $D=11$ supermembrane has remained an elusive challenge since the original analysis performed in \cite{dwhn}. It is well-known that the Hamiltonian of the theory formulated on a Minkowski spacetime in the Light Cone Gauge \cite{dwhn} is the sum of two components. One of the components is associated with the kinematics of the center of mass of the supermembrane described in terms of the zero modes. The other component is the mass operator of the supermembrane which only depends on the non-zero modes.

The ground state wavefunction, $\Psi$, factorizes into two parts
\[
\Psi= \Psi^0\Psi^{\mathrm{non-zero}}.
\]
 The zero mode wavefunction, $\Psi^0$, is responsible for the planar wave associated to the supergravity supermultiplet. The non-zero mode wavefunction, $\Psi^{\mathrm{non-zero}}$, should be annihilated by the mass operator of the supermembrane and should be a singlet under $SO(9)$. This ensures that the full wavefunction $\Psi$ is the unique solution constructed in terms of the $D=11$ supergravity multiplet.

Rigourous treatments of the spectrum (in particular the ground state of the supermembrane) have been achieved by means of an $SU(N)$ regularization of the theory \cite{dwhn}, \cite{goldstone}-\cite{dwmn}. These always involve the quantum mechanics of an $SU(N)$ matrix model which was originally introduced in a different context in \cite{halpern}, \cite{flume} and \cite{baake}. The corresponding Hamiltonian is the starting point of the matrix model theory developed in \cite{bfss}.

The regularized $D=11$ supermembrane was rigorously shown to have a continuous spectrum, the segment from zero to infinity, in \cite{dwln}. The compactification in a sector of the target space by itself does not change this property \cite{dwpp}. However, the spectrum becomes purely discrete with finite multiplicities, when the maps describing the regularized Hamiltonian satisfy a topological condition \cite{mrt} corresponding to a non-trivial central charge in the supersymmetric algebra. See \cite{gmr}-\cite{bgmr3} for a rigourous treatment in this respect. The setting developed in \cite{bgmr3} also shows that the BMN matrix model \cite{bmn} has purely discrete spectrum when considered beyond its semi-classical limit. This argument extended the semiclassical analysis performed in \cite{bmn}.

In this paper we will consider the $SU(N)$ regularized $D=11$ Supermembrane without any topological restriction. The mass operator or Hamiltonian has then continuous spectrum from $0$ to infinity and a problem of great interest is to determine whether $0$ is an eigenvalue of the Hamiltonian or not. In other words, proving or disproving the existence of ground state wave functions.
 
A plausible programme to examine this problem can be divided into three main steps. 
\begin{enumerate}
\item \label{parta} Consider first the Dirichlet problem on a bounded domain $\Omega$ with smooth boundary. In the case of the regularized $SU(N)$ Supermembrane one may consider a ball of arbitrary but finite radius $R$. The domain $\Omega$  remains invariant under the action of the local $SU(N)$ symmetry and global $SO(9)$ symmetry. There the Dirichlet problem may be formulated as \eqref{pro0} below. 
This formulation breaks explicitly supersymmetry as no periodicity condition on the datum at the boundary are imposed at this stage. Since the supersymmetry generates translations, one has to impose periodic boundary conditions in order to define a supersymmetric invariant problem (cf. \cite{mrt}, \cite{gmr}). The solution of the Dirichlet problem is not annihilated by the supercharges, however it minimizes among all wave functions of the Hilbert space satisfying the constraint and the given datum at the boundary in a semi-norm constructed from the supersymmetric charges. 
\item
Solve the external Dirichlet problem\footnote{ Another strategy,  already explored from a numerical perspective    \cite{bgmr3}, is to obtain conditions under which one gets a convergent sequence of solutions when the diameter of $\Omega$ increases to infinity. See the acknowledgements.}. That is, the Dirichlet problem for the mass operator or Hamiltonian on the complement of the domain $\Omega$. The problem should now be formulated on an unbounded region. It is very important here to handle the domain of the Schr\"odinger operator, taking into account the fact that the bosonic potential is unbounded and has valleys extending to infinity where the wavefunction might not vanish at infinity. There are strong results valid on bounded domains (use in the present work) which are not valid for unbounded regions in general. 
 
\item
One must then combine the previous two steps. To this end one must find a patching of the two solutions. One has to choose a suitable boundary datum in order to glue these two solutions with sufficient smoothness. If the smooth patching exists, then there exists the ground state wavefunction of the $D=11$ Supermembrane. Otherwise it does not exist. If the former alternative holds, then one will recover the invariance under supersymmetry of the ground state wave function. Moreover, for this particular boundary data the solution of the first and second stages coincide with the ground state wavefunction of the $D=11$ Supermembrane. 
\end{enumerate}
If the existence and uniqueness of the ground state wave function can be demonstrated by following these three steps,  one can  implement afterwards perturbation techniques in order to analyze properties of the solution.

%%%%%%%%%%%%%%%%%%%%%%%%%%%%%%%%%%%%%%%%%%%%%%%%%%%%%%%%%%%%%%%%%%%%%%%%%%%
\subsection*{Aims and scope of the present work}

 In this paper we will only address part \ref{parta} in the programme described above. In turns we describe a rather general methodology for examining the ground state wavefunctions of a class of supersymmetric models on compact regions subject to Dirichlet boundary conditions. Preliminary results in this direction were already considered in \cite{octonions} and \cite{su2}.

 We focus on the $SU(N)$ regularized supermembrane theory. The corresponding mass operator is the Hamiltonian of the $\mathcal{N}=16$ supersymmetric $SU(N)$ matrix model \cite{bfss}. The center of mass is allowed to move freely in a $D=11$ Minkowski spacetime but the membrane excitations are restricted to a bounded domain $\Omega$ of dimension $9(N^2-1)$ with a smooth boundary $\partial\Omega$.

In main contribution below we establish the existence and uniqueness of the ground state wavefunction of the mass operator, assuming that its values are known on $\partial\Omega$. The mass operator is subject to the physical constraints of the theory which, in the regularized model, generate local $SU(N)$ invariance.  In the large $N$ limit, these constraints are associated to the residual area preserving diffeomorphisms symmetry of the $D=11$ supermembrane in the Light Cone Gauge.  Our arguments depend crucially on the supersymmetric structure of the mass operator.

Mathematically, we solve the homogeneous boundary value problem on $\Omega$ for the Hamiltonian of the supermembrane, given by $H$ in \eqref{sm} below, subject to inhomogeneous boundary conditions. The latter are represented by means of a datum, $g$, on $\partial \Omega$. We will show that a unique wavefunction $\Psi$ exists such that
\begin{equation} \label{pro0} \tag{I}
\begin{cases}\begin{matrix} H\Psi=0  \\
\varphi^A\Psi=0 \end{matrix} & \textrm{in }  \Omega \\
\Psi=g & \textrm{on }   \partial\Omega.
\end{cases}
\end{equation}
The linear map $\varphi^A$ is the operator associated to the $SU(N)$ constraint where $A$ is an index in $SU(N)$, see \eqref{psia} below. Here and elsewhere $g\ne 0$, as  $g=0$ renders $\Psi=0$.

The Hamiltonian $H$ is a selfadjoint elliptic operator on the domain of homogeneous Dirichlet boundary conditions. Hence it has a basis of eigenfunctions vanishing at $\partial \Omega$ and only positive eigenvalues which accumulate at infinity. The ground eigenvalue of this Hamiltonian is therefore positive. Note that the latter is not directly related to the ground state problem on unbounded domains. We call $\Psi$ the ground state wavefunction of the mass operator in $\Omega$, since it corresponds to the restriction to $\Omega$ of the ground state wavefunction of the mass operator in $\mathbb{R}^{9(N^2-1)}$ for $g\not=0$. We will show in section~\ref{5.2} that, remarkably, $g$ can be chosen so that $\Psi$ is invariant under $SO(9)$.

The wavefunction $\Psi$ is the state minimizing a seminorm, among all other states satisfying the constraint and the boundary condition. This seminorm is the one associated to the inner product defined in terms of the supersymmetric charges, $Q$ and $Q^{\dagger}$, given by
\begin{equation} \label{seminormsupercharge}
(Q\eta,Q\lambda)_{\cL_2(\Omega)}+(Q^{\dagger}\eta,Q^{\dagger}\lambda)_{\cL_2(\Omega)} \qquad  \qquad \text{for }\eta, \lambda\in \mathcal{H}^1(\Omega).
\end{equation}
In the subspace\footnote{Here $\cH^1_0(\Omega)$  denotes the completion of the space of smooth functions with support a compact subset of $\Omega$, $C^\infty_{\mathrm{c}}(\Omega)$, in the norm of the Sobolev space $\cH^1(\Omega)$.} $\mathcal{H}^1_0(\Omega)$, this expression defines a norm and is directly related to the Hamiltonian $H$ through Green's identity. However it is only a seminorm in the full Sobolev space $\mathcal{H}^1(\Omega)$. This is analogous to the Dirichlet principle in Electrostatics, see Remark~\ref{rem2}.

\subsection*{Structure of the paper}
 We present the specific description of the Hamiltonian $H$ and formulate our main result about the problem \eqref{pro0} in section~\ref{2}. The proof of this result is deferred to section~\ref{5}. In sections~\ref{3} and \ref{4} we formulate a general framework which we believe is applicable to a large variety of supersymmetric models, where the Hamiltonian is a Schr{\"o}dinger operator with a polynomial potential satisfying conditions of strong ellipticity. In section~\ref{3} we consider existence and uniqueness of the solution for models with global symmetry only. In section~\ref{4} we formulate the criterion for existence and uniqueness of solutions for the more general case of models with gauge symmetry, which requires several technical considerations. In section~\ref{5} we implement the general framework of section~\ref{4}, in order to show existence and uniqueness of the solution of \eqref{pro0}.
We consider the case of the complete regularized supermembrane theory including the regularized area preserving constraint. In the final part of section~\ref{5} we discuss its properties under the $SO(9)$ symmetry.

%%%%%%%%%%%%%%%%%%%%%%%%%%%%%%%%%%%%%%%%%%%%%%%%%%%%%%%%%%%%%%%%%%%%%%%%%%%%%%%%%%%%%%%%%%%%%%%%%%%%%%%%%

\section{Formulation of the problem} \label{2} The $D=11$ supermembrane is described in terms of the membrane coordinates $X^m$ and the Grassmann coordinates $\theta _{\alpha}$. The former corresponds to a vector and the latter to a Majorana spinor. They transform as scalars under diffeomorphisms on the base manifold. The supermembrane theory in the Light Cone Gauge is invariant under rigid supersymmetry, rigid $SO(9)$ symmetry and also under area preserving diffeomorphisms of the base manifold. The latter are the residual gauge symmetry obtained from the original invariance of the action under supermembrane worldvolume diffeomorphisms, once the Light Cone Gauge condition has been imposed.

In \cite{dwhn} the Hamiltonian and the wavefunction were given according to the  symmetry group $SO(9)$, so the above representation of the fields is explicit. For convenience the Majorana spinor is represented by means of linear combinations of elements of the subgroup $SO(7)\times U(1)$. In this way
an explicit expansion $\lambda_{\alpha}$ of the operator associated to the fermionic coordinates in terms of a unique complex spinor of eight components was obtained in \cite{dwhn}. For this purpose, one defines two eigenspinors of $\gamma_9$, called $\theta^{\pm}$, such that
\[\gamma_9\theta^{\pm}=\pm\theta^{\pm}.\]
Then a complex $SO(7)$ spinor satisfies
\[\lambda^{\dag}=2^{1/4}(\theta^+-i\theta^{-})\quad \text{and} \quad \lambda=2^{1/4}(\theta^++i\theta^{-}),\]
where  $\lambda^{\dag}$ is the fermionic conjugate momentum to $\lambda$.

Similarly, the bosonic coordinates $X^{M}$ can be expressed in terms of the representations of $SO(7)\times U(1)$ by means of $(X^m,Z,\overline{Z})$. Here $X^m$ for $m=1,\dots,7$ are the components of a vector in $SO(7)$, and the complex scalars
\[Z=\frac{1}{\sqrt{2}}(X^8+iX^9)\quad \text{and} \quad \overline{Z}=\frac{1}{\sqrt{2}}(X^8-iX^9)\]
transform under $U(1)$. The corresponding bosonic  canonical momenta decouple, as a vector in $SO(7)$ of components $P_m$, a complex momentum  $\mathcal{P}$ in $U(1)$ and its conjugate $\overline{\mathcal{P}}$. That is $P_{M}=(P_{m},\mathcal{P},\overline{\mathcal{P}})$ where
 \[\mathcal{P}=\frac{1}{\sqrt{2}}(P^8-iP^9) \quad \text{and} \quad \overline{\mathcal{P}}=\frac{1}{\sqrt{2}}(P^8+iP^9).\]

Once the theory is regularized by means of the group $SU(N)$ \cite{dwhn},\cite{goldstone}-\cite{dwmn}, the field operators are labeled by an index $A$ in $SU(N)$.  The fields transform in the adjoint representation of the group.

We consider two realizations of the wavefunctions. One of these will be used in the arguments concerning the existence and uniqueness of the ground state wavefunction under an assumption on the kernel of the susy charges. The other one will be used in the proof of this assumption for the $D=11$ supermembrane.

For the first representation, we consider the fermion Fock space. That is a linear space of dimension $2^{8(N^2-1)}$ which carries an irreducible representation of the Clifford algebra generated by $(\lambda^{\dag}+\lambda)$ and $i(\lambda^{\dag}-\lambda)$. The Hilbert space of physical states consists of the wavefunctions which take values in the fermion Fock space and satisfy the first class constraint.

In the second representation the wavefunction comprises elements of a Grassmann algebra generated by $\lambda_{\alpha}^A$ and is given by
\[
\Psi(X_i^A, Z^A,\overline{Z}^A, \lambda_{\alpha}^A)=\sum_{u=0}^{8 (N^2-1)}\Phi_{A_{1}\dots A_{u}}^{\alpha_{1}\dots\alpha_{u}}(X,Z,\overline{Z})\lambda_{\alpha_1}^{A_{1}}
\lambda_{\alpha_2}^{A_{2}}\dots \lambda_{\alpha_u}^{A_{u}}.
\]
In the Schr{\"o}dinger picture, $\lambda_{\alpha A}^{\dag}=\frac{\partial}{\partial\lambda_{\alpha}^A}$ is the  momentum conjugate to $\lambda_{\alpha}^A$.
The coefficient functions $\Phi(X,Z,\overline{Z})$ lie in the usual $\cL_2$ space and the norm of the state is given by
\[
\left\|\Psi\right\|_{\cL_2(\Omega)}^2=\sum_{u=0}^{8(N^2-1)}\frac{1}{u!}\left\|\Phi_{A_{1}\dots A_{u}}^{\alpha_{1}\dots\alpha_{u}}\right\|_{\cL_2(\Omega)}^2.
\]

In \cite{dwhn} it was shown that the zero mode states transform under $SO(9)$ as a $[(44\oplus 84)_{\mathrm{bos}}\oplus 128_{\mathrm{fer}}]$ representation which corresponds to the massless $D=11$ supergravity supermultiplet. Then the construction of the ground state wavefunction reduces to finding a non-trivial solution to
\begin{equation}\label{A}H\Psi=0\end{equation}
where $H=\frac{1}{2}M$ and  $\Psi\equiv \Psi^{\mathrm{non-zero}}$ is required to be a singlet under $SO(9)$. Here $M$ is the mass operator of the supermembrane.

From the supersymmetric algebra, it follows that the Hamiltonian can be express in terms of the supercharges as
\[H =\frac{1}{16}\{Q_{\alpha},Q^{\dagger}_{\alpha}\}.\]
The physical subspace of solutions is given by the kernel of the first class constraint $\varphi^A$. That is
\begin{equation} \label{constraint} \varphi^A\Psi=0\end{equation}
for all $A=1,\dots, N^2-1.$

The supercharges associated to  modes invariant under $SO(7)\times U(1)$ are given explicitly \cite{dwhn} by
\[
\begin{aligned}Q_{\alpha}&=\left\{-i\Gamma_{\alpha\beta}^i\partial_{X_i^A}+\frac{1}{2}f_{ABC}X_i^B X_j^C \Gamma^{ij}_{\alpha\beta}-f_{ABC} Z^B\overline{Z}^C\delta_{\alpha\beta}\right\}\lambda_{\beta}^A\\
&+\sqrt{2}\left\{\delta_{\alpha\beta}\partial_{Z^A}+i f_{ABC}X_i^B \overline{Z}^C \Gamma^i_{\alpha\beta} \right\}\partial_{\lambda_{\beta}^A}\end{aligned}\]
and
\[
\begin{aligned}Q_{\alpha}^{\dagger}&=\left\{i\Gamma_{\alpha\beta}^i\partial_{X_i^A}+\frac{1}{2}f_{ABC}X_i^B X_j^C \Gamma^{ij}_{\alpha\beta}+f_{ABC} Z^B\overline{Z}^C\delta_{\alpha\beta}\right\}\partial_{\lambda_{\beta}^A}\\
&+\sqrt{2}\left\{-\delta_{\alpha\beta}\partial_{\overline{Z}^A}+i f_{ABC}X_i^B Z^C \Gamma^i_{\alpha\beta} \right\}\lambda_{\beta}^A.\end{aligned}\]
The corresponding superalgebra satisfies \cite{dwhn}
\[\begin{aligned}&\{Q_{\alpha},Q_{\beta}\}=2\sqrt{2}\delta_{\alpha\beta}\overline{Z}^A\varphi_A,\\
&\{Q_{\alpha}^{\dagger},Q_{\beta}^{\dagger}\}=2\sqrt{2}\delta_{\alpha\beta}Z^A\varphi^A,\\
&\{Q_{\alpha},Q^{\dagger}_{\beta}\}=2\delta_{\alpha\beta}H-2i\Gamma^i_{\alpha\beta}X_i^A\varphi_A.\end{aligned}\]
The Hamiltonian associated to the the regularized mass operator of the supermembrane \cite{dwhn} is then
\begin{equation} \label{sm}
\begin{aligned}
H&=\frac{1}{2}M=-\Delta+V_{\mathrm{B}}+V_{\mathrm{F}}\\
\Delta&=\frac{1}{2}\frac{\partial^2}{\partial X^i_A\partial X_i^A}+\frac{1}{2}\frac{\partial^2}{\partial Z_A\partial \overline{Z}^A}\\
V_{\mathrm{B}}&=\frac{1}{4}f_{AB}^Ef_{CDE}\{X_i^AX_j^BX^{iC}X^{jD}+4X_i^AZ^BX^{iC}\overline{Z}^{D}+2Z^A\overline{Z}^B\overline{Z}^{C}Z^{D}\}\\
V_{\mathrm{F}}&=if_{ABC}X_i^A\lambda_{\alpha}^B\Gamma_{\alpha\beta}^i\frac{\partial}{\partial\lambda_{\beta C}}
+\frac{1}{\sqrt{2}}f_{ABC}(Z^A\lambda_{\alpha}^B\lambda_{\alpha}^C-\overline{Z}^A\frac{\partial}{\partial \lambda_{\alpha B}}\frac{\partial}{\partial\lambda_{\alpha C}}).
\end{aligned}
\end{equation}
The generators of the local $SU(N)$ symmetry are
\begin{equation} \label{psia}
\varphi^A =f^{ABC}\left( X_{i}^B\partial_{X_i^C}+Z_B\partial_{Z^C}+\overline{Z}_B\partial_{\overline{Z}^C}
+\lambda_{\alpha}^{B}\partial_{\lambda_{\alpha}^C}\right).
\end{equation}
They annihilate the physical states. They commute with $Q_{\alpha}$ and $Q_{\alpha}^{\dagger}$, hence also with $H$.

The Hamiltonian $H$ is a positive operator. It annihilates $\Psi$ on the physical subspace, if and only if  $\Psi$ is a singlet under supersymmetry. In this case,
$$Q_{\alpha}\Psi=0\quad \text{and}\quad Q_{\alpha}^{\dagger}\Psi=0.$$
The latter ensures that the wavefunction is massless, however it does not guarantee that the ground state wavefunction is the corresponding supermultiplet associated to supergravity. This holds true, only when $\Psi$ is a singlet under $SO(9)$.

The following is the main result of this paper. We strongly believe it provides a valuable insight into the problem of existence for the ground state of $H$ on an unbounded domain. In particular it settles completely part \ref{parta} of the programme mention in section~\ref{intro}.

\begin{theorem} \label{propermain}
For $\Omega$ a $9(N^2-1)$-dimensional ball and $g\not=0$ sufficiently regular in $\Omega$, the boundary value problem \eqref{pro0} always has a solution $\Psi$ and this solution is always unique.
\end{theorem}

The proof and the precise assumptions on the regularity of $g$ here will be discussed in section~\ref{5}. See Theorem~\ref{main}. Under suitable conditions on $g$ the solution is a singlet under $SO(9)$.

%%%%%%%%%%%%%%%%%%%%%%%%%%%%%%%%%%%%%%%%%%%%%%%%%%%%%%%%%%%%%%%%%%%%%%%%%%%%%%%%%%%%%%%%%%%%%%%%%%%%%%%%%%%%%%
%%%%%%%%%%%%%%%%%%%%%%%%%%%%%%%%%%%%%%%%%%%%%%%%%%%%%%%%%%%%%%%%%%%%%%%%%%%%%%%%%%%%%%%%%%%%%%%%%%%%%%%%%%%%%

%%%%%%%%%%%%%%%%%%%%%%%%%%%%%%%%%%%%%%%%%%%%%%%%%%%%
\section{Matrix models with global symmetries}  \label{3}

  In this section we describe a setting for the analysis of ground states which applies to a wide variety of matrix models on compact domains, given appropriate boundary data. The presence of an area preserving constraint will be address in section~\ref{4}.

Below $\Omega\subset \mathbb{R}^{9(N^2-1)}$ will be a bounded open set whose boundary
$\partial\Omega$ is of class $C^\infty$. Here and elsewhere $\cL_2(\Omega)$, $\cH^1(\Omega)$ and $\cH^2(\Omega)$, are the corresponding Lebesgue and Sobolev Hilbert spaces of fields in $\Omega$ with $d= 9(N^2-1)$ components. We will denote the inner product of $\cL_2(\Omega)$ by $(\cdot,\cdot)_{\cL_2(\Omega)}$. Here $d$ is usually a large integer. The space $\cH^1_0(\Omega)$  is the completion in the norm of $\cH^1(\Omega)$ of the subspace $C^\infty_{\mathrm{c}}(\Omega)$, the functions ($d$ components also) with support a compact subset of $\Omega$.

\subsection*{Conditions on a generic Hamiltonian}
Let
\[\fH=-\nabla^2\mathbb{I}+V\]
be a Schr{\"o}dinger operator
where the matrix potential $V=V^\dagger\in C^{\infty}(\overline{\Omega})$. Then $\fH$ is strongly elliptic \cite[Chapter 7]{Folland}. Let
\[
\operatorname{Dom}(\fH)=\cH^{1}_{0}(\Omega)\cap \cH^2(\Omega).
\]
Since $V:\cL_2(\Omega)\longrightarrow \cL_2(\Omega)$ is bounded and symmetric, according to Kato-Rellich's theorem,
\[
     \fH:\operatorname{Dom}(\fH)\longrightarrow \cL_2(\Omega)
\]
is a selfadjoint operator.

Everywhere below we assume that $\fH$ is the Hamiltonian of a supersymmetric theory in the following precise sense. The identity
\begin{equation}   \label{susycondition} \tag{S}
    \fH= \frac{1}{2}(Q Q^{\dag}+Q^{\dag}Q)=[Q,Q^\dag]
\end{equation}
holds true, for $Q$ a supercharge operator. This supercharge operator is a first order differential operator satisfying supersymmetric superalgebra conditions, and such that
\[
Q,Q^{\dag}:\cH^1(\Omega) \longrightarrow \cL_2(\Omega)
\]
and
\[
    Q,Q^{\dag}:\cH^2(\Omega)\cap \cH^1_0(\Omega) \longrightarrow \cH^1(\Omega).
\]
The latter ensures that the domain of $\fH$ in the representation \eqref{susycondition} is mapped appropriately.

Additionally we will assume that this supercharge operator satisfies the condition
\begin{equation}    \label{superchargecond} \tag{K}
    \ker(Q|_{\cH^1_0(\Omega)})\cap \ker(Q^\dag|_{\cH^1_0(\Omega)})=\{0\}.
\end{equation}
That is
\[
    Q\psi=0 \text{ and } Q^\dag\psi=0\text{ for }\psi\in \cH^1_0(\Omega)
    \quad \Rightarrow \quad \psi=0.
\]
This will be satisfied by the supermembrane Hamiltonian $H$, also by the Hamiltonian considered at the end of section~\ref{3}, and it is also true for other interesting cases \cite{su2,octonions}.

All the results reported here depend strongly on the condition \eqref{superchargecond}. The following lemma highlights the role played by this assumption in relation to the boundary value problem associated to $\fH$. The identity \eqref{hkernill} is crucial for the existence and uniqueness of the solutions of \eqref{pro0}.

\begin{lemma} \label{nullkernel}
Let $\psi\in\operatorname{Dom}(\fH)$ be such that $\fH\psi=0$. Then
$\psi=0$.
\end{lemma}
 \begin{proof}
If $\psi$ is as in the hypothesis, then
\begin{equation}  \label{hkernill}
\begin{aligned}
0=(\psi, \fH\psi)&=\frac{1}{2}(\psi,QQ^{\dag}\psi)+\frac{1}{2}(\psi,Q^{\dag}Q\psi)  \\ &=\frac12\left(\left\|Q^{\dag}\psi\right\|^2_{\cL_2(\Omega)}+\left\|Q\psi\right\|^2_{\cL_2(\Omega)}\right).
\end{aligned}
\end{equation}
Hence $\psi\in \ker(Q|_{\cH^1_0(\Omega)})\cap \ker(Q^\dag|_{\cH^1_0(\Omega)})$ and $\psi=0$ according to \eqref{superchargecond}.
\end{proof}

That is, \eqref{superchargecond} implies
\[
     \ker{\fH}=\{0\}
\]
on a supersymmetric Hamiltonian \eqref{susycondition}.

\subsection*{The Dirichlet form}
Define the strongly elliptic Dirichlet form (of order one) associated to $\fH$ by means of the identity
\[
\fD(\phi, \psi)=(\nabla\phi,\nabla\psi)+(\phi,V\psi).
\]
 Then
\[
    \fD:\cH^1_0(\Omega)\times \cH^1_0(\Omega)\longrightarrow \mathbb{C}
\]
is a non-negative coercive closed quadratic form.

That $\fD$ is coercive, means that
for suitable constants $C>0$ and $\lambda\geq 0$ (sufficiently large),
\[
\fD(\psi,\psi)\ge C\left\|\psi\right\|^2_{\cH^1(\Omega)}-\lambda\left\|\psi\right\|^2_{\cL_2(\Omega)}
\qquad \text{for all } \psi\in \cH^1_0(\Omega).
\]
In the current setting, the inequality is valid for $C=1$ and $\lambda=1+\Lambda$, where $\Lambda$ is a lower bound of the minimum eigenvalue of the potential on $\overline{\Omega}$.

That the form $\fD$ is non-negative can be seen as follows.
For all $\psi\in \operatorname{Dom}(\fH)$,
\begin{equation} \label{diriform}
     \fD(\psi,\psi)=\|Q\psi\|_{\cL_2(\Omega)}^2+\|Q^\dag \psi\|_{\cL_2(\Omega)}^2\geq 0.
\end{equation}
Recall \eqref{hkernill}.
As $\operatorname{Dom}(\fH)$ is a core (in the form sense) for
$\fD$, then \eqref{diriform} also holds true for all $\psi\in \cH^1_0(\Omega)$.

By virtue of Lemma~\ref{nullkernel}, it then follows that the ground eigenvalue of $\fH$ is strictly positive. Note that any $\psi\in \operatorname{Dom}(\fH)$ vanishes on $\partial\Omega$.

\subsection*{The boundary value problems}
The Dirichlet problems associated to $\fH$ can be re-written in terms of $\fD$. Here we assume that data, $g\in \cH^2(\Omega)$ and $f=(\nabla^2-V)g\in \cL_2(\Omega)$, are given.

The  homogeneous Dirichlet problem with inhomogeneous boundary conditions associated to $\fH$ is formulated as follows. Find $\Psi\in \cH^2(\Omega)$ such that
\begin{equation} \label{pro1} \tag{II}
\begin{cases}(-\nabla^2+V)\Psi=0 & \textrm{in }  \Omega\\
\Psi=g & \textrm{on }   \partial\Omega.
\end{cases}
\end{equation}
This is related to the inhomogeneous Dirichlet problem with homogeneous boundary conditions.
Find $\Phi\in \cH^2(\Omega)$ such that
\begin{equation} \label{pro2}   \tag{III}
\begin{cases}
(-\nabla^2+V)\Phi=f\quad & \textrm{in } \Omega \\
\Phi=0 & \textrm{on } \partial\Omega.
\end{cases}
\end{equation}
The weak formulation of \eqref{pro2} is the weak inhomogeneous Dirichlet problem with homogeneous boundary conditions.
Find $\Phi\in \cH^1_0(\Omega)$ such that
\begin{equation} \label{pro3} \tag{IV}
\fD(\phi,\Phi)=(\phi,f)  \textrm{ for all } \phi\in \cH^1_0(\Omega).
\end{equation}

By setting $\Psi=\Phi+g$ we see that \eqref{pro1} and \eqref{pro2} are equivalent.
Clearly a solution of \eqref{pro2} would also be a  solution of \eqref{pro3}. Normally the latter is also called a weak solution.
Moreover, a solution $\Phi$ of \eqref{pro3} originally in $\cH^1_0(\Omega)$ will also be in $\cH^2(\Omega)$ and would satisfy \eqref{pro2} (see below for the precise statement).  In passing from \eqref{pro1} to \eqref{pro3} the boundary condition $\Psi=g$ on  $\partial\Omega$ has been replaced by the condition $\Phi\in \cH^1_0(\Omega)$.

Let $\delta(\xi)$ be the seminorm associated to the inner product \eqref{seminormsupercharge},
\begin{equation*}
\delta^2(\xi)=\left\|Q\xi\right\|^2_{\cL_2(\Omega)}+\left\|Q^\dagger\xi\right\|^2_{\cL_2(\Omega)},
\end{equation*}
for $\xi\in \cH^1(\Omega)$. The solution $\Psi$ to the inhomogeneous Dirichlet problem \eqref{pro1} is the state that minimises this seminorm, among all other states $\xi\in \cH^1(\Omega)$ satisfying the same boundary condition $\xi=g$ on $\partial\Omega$.
Indeed, \begin{equation*}
\xi-\Psi=\eta \in \cH_0^1(\Omega)
\end{equation*}
and
\begin{equation*}
\delta^2(\xi)=\delta^2(\Psi)+\delta^2(\eta)\geq \delta^2(\Psi),
\end{equation*}
since
\begin{equation*}
(Q\Psi,Q\eta)_{\cL_2(\Omega)}+(Q^{\dagger}\Psi,Q^{\dagger}\eta)_{\cL_2(\Omega)}= ((-\nabla^2+V)\Psi,\eta)_{\cL_2(\Omega)}=0.
\end{equation*}
Below we will show that this minimising state $\Psi$ exists and is unique. Note that $\delta^2(\eta)=0$ implies $\eta=0$, because $\eta\in \cH^1_0(\Omega)$.

\begin{remark} \label{rem2}
The electrostatic field in the vacuum fulfils an analogous property. The electrostatic energy $E$ of an electrostatic potential $\xi$ on a bounded domain $\Omega$ is given by
\begin{equation*}
E=\int_{\Omega}\nabla\xi\nabla\xi.
\end{equation*}
This is a seminorm in $\cH^1(\Omega)$. The harmonic potential is the one that minimises $E$ among all other potentials satisfying the same boundary condition on $\partial\Omega$.
\end{remark}
%%%%%%%%%%%%%%%%%%%%%%%%%%%%%%%%%%%%%%%%%%%%%%%%%%%%%%%%%%%%%%%%%%%%%%%%%

%%%%%%%%%%%%%%%%%%%%%%%%%%%%%%%%%%%%%%%%%%%%%%%%%%%%%%%%%%%%%%%%%%%%%%%%%%%%%%%%%%%%%%%%%%%%%%%%%%%
\subsection*{Existence and uniqueness of solutions}

As we shall see next, for regular data as above, \eqref{pro1} and \eqref{pro2} are always uniquely solvable.

\begin{lemma} \label{existunique}
Let $g\in \cH^2(\Omega)$. There always exists a unique solution $\Phi\in \cH^1_0(\Omega)\cap \cH^2(\Omega)$ of \eqref{pro2} and a corresponding unique solution  $\Psi= \Phi+g\in \cH^2(\Omega)$ of \eqref{pro1}.
\end{lemma}
\begin{proof}
Recall that $\fD$ is a non-negative symmetric Dirichlet form of order~$1$. Let
\[
\mathcal{K}=\{\xi\in \cH_0^1(\Omega): \fD(\varphi,\xi)=0\quad\textrm{for all}\quad \varphi\in \cH_0^1(\Omega)\}.
\]
By virtue of regularity results for strongly elliptic Dirichlet forms \cite[Theorem~(7.32)]{Folland}, it follows that $\mathcal{K}\subset \cH^2(\Omega)$ and in fact $\mathcal{K}=\ker(\fH)$. Then, according to
Lemma~\ref{nullkernel}, $\mathcal{K}=\{0\}$.
Hence, by virtue of \cite[Theorem (7.21)]{Folland}, there exists $\Phi\in \cH_0^1(\Omega)$ such that  \eqref{pro3} holds true. Moreover, \cite[Theorem~(7.32)]{Folland} in fact $\Phi\in \cH^{2}(\Omega)$. Thus $\Phi$ is also a solution to \eqref{pro2} and $\Psi= \Phi+g$ a solution to \eqref{pro1}.

As $\mathcal{K}=\{0\}$, it immediately follows that the solution is unique.
\end{proof}
%%%%%%%%%%%%%%%%%%%%%%%%%%%%%%%%%%%%%%%%%%%%%%%%%%%%%%%%%%%%%%%%%%
\subsection*{Pointwise regularity at the boundary}
A crucial observation on the regularity properties of supercharge operators of first order is now in place. This observation is independent of the assumption \eqref{superchargecond}, however, for its validity, the potential should be smooth.

\begin{lemma}   \label{regularity}
Let $\tilde{Q}:\cH^1(\Omega)\longrightarrow \cL_2(\Omega)$ be a (generic) supercharge operator of first order.  If $\tilde{Q}\Phi=0$ and  $\tilde{Q}^{\dag}\Phi=0$ for $\Phi\in \cH^1_0(\Omega)$,
then $\Phi\in C^\infty(\overline{\Omega})$ and
\[
\tilde{Q}\Phi(\mathbf{x})=\tilde{Q}^{\dag}\Phi(\mathbf{x})=0 \qquad \text{for all} \qquad \mathbf{x}\in \overline{\Omega}.
\]
\end{lemma}
\begin{proof}
Let $\Phi$ be as in the hypothesis. By virtue of classical bootstrap arguments and the Sobolev Lemma \cite{Folland}, it follows that
$\Phi\in C^\infty(\overline{\Omega})$.
That is $\Phi$ is smooth in the domain up to the boundary. Thus also $\tilde{Q}\Phi$ and $\tilde{Q}^{\dag}\Phi$ lie in $C^\infty(\overline{\Omega})$. Hence $\tilde{Q}\Phi(\mathbf{x})=\tilde{Q}^{\dag}\Phi(\mathbf{x})=0$ for all $\mathbf{x}\in \overline{\Omega}$.
\end{proof}

%%%%%%%%%%%%%%%%%%%%%%%%%%%%%%%%%%%%%%%%%%%%%%%%%%%%%%%%%%%%%%
\subsection*{A toy model}
Consider a version of the toy model introduced in \cite{dwln} on a compact region. In \cite{Ghh} it was shown that this model has no zero eigenvalue for the non-compact problem.
By combining the supersymmetric structure of the Hamiltonian shown below with Lemma~\ref{existunique}, it follows that the solutions to the problems \eqref{pro1} and \eqref{pro2} exist and are unique in this case.

 Let
\[\fH=p_x^2+p_y^2+x^2y^2+x\sigma_3+y\sigma_1\]
where $\sigma_i$ are the Pauli matrices. The supersymmetric charges in this case are
\begin{equation*}
Q=Q^{\dag}=
\begin{pmatrix}
-xy & i\partial_x-\partial_y\\
i\partial_x-\partial_y & xy
\end{pmatrix}.
\end{equation*}
The wavefunctions are such that
\begin{equation}
\Phi=
\begin{pmatrix}
\Phi_1 \\
\Phi_2
\end{pmatrix} \qquad \text{and} \qquad \Phi=0\quad\textrm{on}\quad \partial\Omega.
\end{equation}

We firstly show that the condition \eqref{superchargecond} is valid for the supersymmetric charges. Let $\Phi\in \cH^1_0(\Omega)$ be such that
\begin{equation} \label{vanishingboth}
     Q\Phi=Q^\dag\Phi=0 \quad \textrm{in}\quad \Omega.
\end{equation}
According to Lemma~\ref{regularity}, this condition is satisfied pointwise up to the boundary of $\Omega$. Let $\mathbf{x}\in \partial\Omega$ and denote by $\mathbf{n}_1,\mathbf{n}_2$ the components of the normal to $\partial\Omega$ at $\mathbf{x}$. The tangent to $\partial\Omega$ at $\mathbf{x}$ is then $(\mathbf{n}_2,-\mathbf{n}_1)$, and we must have
\[(\mathbf{n}_2\partial_x-\mathbf{n}_1\partial_y)\Phi(\mathbf{x})=0.\]

The solutions are regular, so we can extend them continuously up to the boundary. Then \eqref{vanishingboth} yields
\[(i\partial_x+\partial_y)\Phi_2(\mathbf{x})=(i\partial_x-\partial_y)\Phi_1(\mathbf{x})=0\]
pointwise. Since $(\mathbf{n}_1,\mathbf{n}_2)\neq 0$, if $\mathbf{n}_2\ne 0$
\begin{equation}
\left\{\begin{aligned}
&(1+i\frac{\mathbf{n}_1}{\mathbf{n}_2})\partial_y\Phi_2(\mathbf{x})=0&\Rightarrow\quad \partial_y\Phi_2(\mathbf{x})=0\textrm{ and }\partial_x\Phi_2(\mathbf{x})=0\\
&(-1+i\frac{\mathbf{n}_1}{\mathbf{n}_2})\partial_y\Phi_1(\mathbf{x})=0 &\Rightarrow\quad \partial_y\Phi_1(\mathbf{x})=0\textrm{ and } \partial_x\Phi_1(\mathbf{x})=0.
\end{aligned} \right.
\end{equation}
A similar conclusion is obtained for $\mathbf{n}_1\ne 0$. Hence $\Phi$ and $\partial_{\mathbf{n}}\Phi$ must vanish on  $\partial\Omega$.
According to the Cauchy-Kovalevskaya Theorem \cite{Folland}, from the fact that the potential is analytic, we conclude that $\Phi=0$  pointwise in $\overline{\Omega}$. This yields \eqref{superchargecond}.

Similar arguments can be employed in order to derive existence and uniqueness of the solution for the $SU(N)$ truncated model of the $D=11$ supermembrane considered in \cite{dwhn}, which only has global symmetries. See \cite{octonions}.

%%%%%%%%%%%%%%%%%%%%%%%%%%%%%%%%%%%%%%%%%%%%%%%%%%%%%%%%%%%%%%%%%%%%%%%%%%%D=11%
%%%%%%%%%%%%%%%%%%%%%%%%%%%%%%%%%%%%%%%%%%%%%%%%%%%%%%%%%%%%%%%%%%%%%%%%%%%D=11%

\section{Systems with gauge symmetry} \label{4}
We now consider supersymmetric theories subject to a constraint. This constraint realises as a subspace decomposition which diagonalizes the Hamiltonian and its Dirichlet form, while it remains compatible with the boundary conditions.

\pagebreak

\begin{assumption}[generic constraint] \label{asum} \
\begin{enumerate}
\item\label{asumb} There exist two subspaces $\X^1_0,\Y^1_0\subseteq \cH^1_0(\Omega)$ such that
\begin{itemize}
\item they are both closed in the norm of $\cH^1(\Omega)$,
\item they are orthogonal to one another in the inner product of $\cL_2(\Omega)$ and
\item $\cH^1_0(\Omega)=\X^1_0+\Y^1_0$.
\end{itemize}
We denote by  $\X^0,\Y^0\subseteq \cL_2(\Omega)$, respectively, the completion of $\X^1_0$ and $\Y^1_0$ in the norm of $\cL_2(\Omega)$.
We write $\X^2=\cH^2(\Omega)\cap \X^0$ and $\Y^2=\cH^2(\Omega)\cap \Y^0$.
\item\label{asumd} We have the decomposition
\[\left(\X^2\cap \X^1_0\right)+ \left(\Y^2 \cap \Y^1_0\right)=\cH^2(\Omega)\cap \cH^1_0(\Omega)\] and
there exist two operators
\begin{equation*} \label{factorsofH}
     \fH_{\X}:\X^2\cap \X^1_0 \longrightarrow \X_0 \qquad \text{and} \qquad
 \fH_{\Y}:\Y^2\cap \Y^1_0 \longrightarrow \Y_0
\end{equation*}
satisfying the following. The Hamiltonian decomposes
as \[\fH\psi = \fH_{\X}\psi_{\X}+ \fH_{\Y}\psi_{\Y}
\qquad \text{for all} \quad \begin{cases}
\psi=\psi_\X+\psi_\Y \\ \psi_{\X}\in \X^2\cap \X^1_0 \\
\psi_{\Y}\in \Y^2\cap \Y^1_0. \end{cases}
\]
\end{enumerate}
\end{assumption}

Various consequence can be derived from these assumptions.
Firstly note that
\[
    \cL_2(\Omega)=\X^0+\Y^0.
\]
Also $\X^2$, $\Y^2$, $\X^2\cap \X^1_0$ and $\Y^2\cap \Y^1_0$ are closed subspaces of $\cH^2(\Omega)$, in the corresponding norm. From the definition, it immediately follows that
\[
     \X^2 \subseteq \X^0 \qquad \text{and}  \qquad
   \Y^2 \subseteq \Y^0.
\]
Moreover, the closures of $\X^2$ and $\Y^2$ in the norm of $\cL_2(\Omega)$ are exactly $\X^0$ and $\Y^0$, respectively.
Hence $\X^2+\Y^2$ is dense in $\cL_2(\Omega)$.

We also have the representation
\[
       \X^1_0=\cH^1_0(\Omega)\cap \X^0 \quad \text{and} \quad
       \Y^1_0=\cH^1_0(\Omega)\cap \Y^0.
\]
Alongside with the two conditions in the Assumption~\ref{asum}, this implies that $\fH_\X$ and $\fH_\Y$ are selfadjoint operators in the Hilbert spaces $\X^0$ and $\Y^0$, respectively.

Let
\begin{align*}
    \fD_{\X}=\fD|_{\X^1_0 \times\X^1_0}&:\X^1_0 \times\X^1_0\longrightarrow \mathbb{C}
\\
\fD_{\Y}=\fD|_{\Y^1_0 \times\Y^1_0}&:\Y^1_0 \times\Y^1_0\longrightarrow \mathbb{C}
\end{align*}
be the restrictions of the Dirichlet form to the corresponding subspaces. Then both these forms are closed, symmetric and bounded below. Moreover, since
    \begin{align*}
    \fD_{\X}(\phi_\X,\psi_\X)=( \fH_\X \phi_\X,\psi_\X)_{\cL_2(\Omega)}& \qquad \text{for all }
    \phi_\X\in \X^2\cap \X^1_0 \text{ and }\psi_\X \in \X^1_0
\\
  \fD_{\Y}(\phi_\Y,\psi_\Y)=( \fH_\Y \phi_\Y,\psi_\Y)_{\cL_2(\Omega)}& \qquad \text{for all }
    \phi_\Y\in \Y^2\cap \Y^1_0 \text{ and }\psi_\Y \in \Y^1_0,
\end{align*}
then $\fH_\X$ and $\fH_\Y$ are, respectively, the selfadjoint operators associated to $\fD_{\X}$ and
 $\fD_{\Y}$, via Kato's First Representation Theorem \cite[Th VI.2.1]{Kato}. From this, it follows that
\[
    \fD(\phi,\psi)=\fD_\X(\phi_\X,\psi_\X)+\fD_\Y(\phi_\Y,\psi_\Y)
\]
for all $\phi,\psi\in \cH^1_0(\Omega)$ in the corresponding representations $\phi=\phi_\X+\phi_\Y$ and
$\psi=\psi_\X+\psi_\Y$ provided by \ref{asumb}.

%%%%%%%%%%%%%%%%%%%%%%%%%%%%%%%%%%%%%%%%%%%%%%%%%%%%%%%%%%%%%%
\subsection*{The constrained boundary value problems} In what follows $\X^0$ is the subspace of physical states.
Set data: $g\in \X^2$ and $f=(\nabla^2-V)g\in \X^0$. Consider the constrained versions of \eqref{pro1}, \eqref{pro2} and \eqref{pro3}.
\begin{equation} \label{pro1x} \tag{V}
\begin{cases}(-\nabla^2+V)\Psi=0 & \textrm{in }  \Omega\\
\Psi\in \X^2 \\
\Psi=g & \textrm{on }   \partial\Omega,
\end{cases}
\end{equation}
\begin{equation} \label{pro2x}   \tag{VI}
\begin{cases}
(-\nabla^2+V)\Phi=f\quad & \textrm{in } \Omega \\
\Phi \in \X^2 \\
\Phi=0 & \textrm{on } \partial\Omega
\end{cases}
\end{equation}
and
\begin{equation} \label{pro3x} \tag{VII}
\fD_\X(\phi,\Phi)=(\phi,f)  \textrm{ for all } \phi\in \X^1_0.
\end{equation}
Let $\X^1=\cH^1(\Omega)\cap \X^0$ and $\Y^1=\cH^1(\Omega)\cap \Y^0$.
Then $\X^1$ and $\Y^1$ are closed subspaces of $\cH^1(\Omega)$.
Consider the following condition which is similar but weaker  than \eqref{superchargecond} from the previous section,
\begin{equation}    \label{superchargecondX} \tag{K$_\X$}
Q\psi_\X=Q^\dag\psi_\X=0\text{ for }\psi_\X\in \X^1_0
    \quad \Rightarrow \quad \psi_{\X}=0.\end{equation}
That is, the supercharge operator satisfies an analogue to \eqref{superchargecond}, but only in the constraint subspace $\X$.
Note that \eqref{superchargecond} implies \eqref{superchargecondX}.

\begin{lemma}  \label{existuniqueX}
Suppose that a decomposition as specified in the Assumption~\ref{asum} as well as \eqref{superchargecondX} hold true.
Let $g\in \X^2$ and $f=(\nabla^2-V)g\in \X^0$. There always exists a solution $\Phi\in \X^2\cap \X^1_0(\Omega)$ of \eqref{pro2x} which is unique. Moreover, a corresponding solution  $\Psi= \Phi+g\in \X^2$ of \eqref{pro1x} also exists and is unique.
\end{lemma}
\begin{proof}
We first show that \eqref{pro3x} has a solution.  By
virtue of the assumption \eqref{superchargecondX} and the fact that
\[
      \fD_\X(\phi_\X,\phi_\X)=\fD(\phi_\X,\phi_\X)=\frac12(\|Q \phi_\X\|_{\cL^2(\Omega)}^2+
      \|Q^\dag \phi_\X\|_{\cL^2(\Omega)}^2)
\]
for all $\phi_\X\in \X^1_0$, then $\fD_\X$ is positive and $\X^1_0$ is a Hilbert space
with respect to the corresponding inner product given by this form.
By the Lax-Milgram Theorem \cite[\S 6.2]{1997Evans}, there exists  a solution $\Phi\in \X^1_0$ for \eqref{pro3x}.

Now, since $\Phi\in \X^1_0$, the part of $\Phi$ that lies in $\Y^1_0$ according to the decomposition \ref{asumb} of the Assumption~\ref{asum} is $\Phi_\Y=0$. Hence,
\[
   \fD(\phi,\Phi)=\fD_\X(\phi_\X,\Phi)+\fD_\Y(\phi_\Y,0)=( \phi_\X,f)_{\cL_2(\Omega)} = ( \phi,f)_{\cL_2(\Omega)}
\]
for all $\phi\in \cH^1_0(\Omega)$ represented as $\phi=\phi_\X+\phi_\Y$ for $\phi_\X\in\X^1_0$ and $\phi_\Y\in \Y^1_0$.
Thus $\Phi$ is also a solution of \eqref{pro3}. By repeating the same steps as in the proof of Lemma~\ref{existunique} (which applies on physical states),
we get that $\Phi\in \cH^2(\Omega)\cap \cH^1_0(\Omega)$. Hence,
we have $\Phi\in \X^2\cap \X^1_0$  and so
$\Phi$ is a solution of \eqref{pro2x}.

The rest of the proof follows from a similar argument as the one presented in Lemma~\ref{existunique}.
\end{proof}

%%%%%%%%%%%%%%%%%%%%%%%%%%%%%%%%%%%%%%%%%%%%%%%%%%%%%%%%%%%%%%%%%%%%%%%%

%%%%%%%%%%%%%%%%%%%%%%%%%%%%%%%%%%%%%%%%%%%%%%%%%%%%%%%%%%%%%%%%%%%%%%%%
%%%%%%%%%%%%%%%%%%%%%%%%%%%%%%%%%%%%%%%%%
\section{The ground state of the $D=11$ supermembrane} \label{5}
Consider the supermembrane theory on a $D=11$ Minkowski space-time introduced in  section~\ref{2}. We formulate a precise result which implies the validity of Theorem~\ref{propermain}. The relevance of the present setting is twofold. On the one hand, it is a problem of physical interest by itself due to its potential implications in M-theory. On the other hand, it is a crucial step towards the solution of the ground state problem on an unbounded domain. As we are in the presence of a gauge constraint, we resource to the framework of section~\ref{4}.

Let $H$ be the Hamiltonian given by the expression \eqref{sm}. Let $\Omega$ be a ball in $\mathbb{R}^{9(N^2-1)}$ of radius $R>0$. The boundary, $\partial\Omega$, is a sphere of dimension $9(N^2-1)-1$ with the same radius. Consider the $SO(7)\times U(1)$ decomposition as described in section~\ref{2}.
The coordinates are  $(X_i^A, Z^A,\overline{Z}^A, \lambda_{\alpha}^A)$ where $A$ is the $SU(N)$ index. The radial coordinate $\rho$ is defined by \[\rho^2= (X_i^A)^2+ 2Z\overline{Z}.\]
This radial coordinate and hence $\Omega$, are invariant under the symmetry generated by the first class constraint. That is, the generators of local $SU(N)$ transformations. Consequently the constraint imposes no restriction to the normal derivative on the border, $\partial_{\rho}\Psi\vert_{\partial \Omega}$. The constraints $\varphi^A$ given in section~\ref{2} commute with $Q_{\alpha},Q^{\dagger}_{\alpha}$ and $H$, and \[\varphi^A: \cH^2(\Omega)\cap \cH^1_0(\Omega)\longrightarrow \cH^1_0(\Omega).\]

\subsection*{Validity of \eqref{superchargecond}}
Let us verify the condition \eqref{superchargecond} for the supersymmetric charges. Assume that
\begin{equation}\label{u2}
 Q_{\alpha}\psi=0 \quad \text{ and }\quad Q_{\alpha}^{\dag}\psi=0 \quad \textrm{in} \quad \Omega\end{equation}
 for $\psi\in \cH^{2}(\Omega)\cap \cH_0^1(\Omega)$. We wish to prove that $\psi=0$ in $\Omega$.

Regularity at the boundary (Lemma~\ref{regularity}), allows us to extend the restriction on $\frac{\partial\psi}{\partial \rho^{'}}$ arising from (\ref{u2}) smoothly to the boundary. The conditions $Q_{\alpha}\psi\vert_{\partial\Omega}=0$ and $Q_{\alpha}^{\dag}\psi\vert_{\partial\Omega}=0$ for the $SU(N)$ regularized supermembrane found in \cite{dwhn}, now evaluated on the boundary where $\psi=0$, are
\begin{align}
\label{u3} Q_{8}\psi&=\sqrt{2}\partial_{Z^A}\partial_{\lambda_{8}^A}\psi- i \Gamma_{8j}^i \partial_{X_i^A}\lambda_j^A \psi=0
\\
\label{u4} Q_{8}^{\dag}\psi&=-\sqrt{2}\partial_{\overline{Z}^A}\lambda_{8}^A\psi+ i\Gamma_{8j}^i \partial_{X_i^A}\partial_{\lambda_j^A} \psi=0\\
\label{u5}
Q_{j}\psi&=-i\Gamma_{j8}^i \lambda_8^A\partial_{X_i^A}\psi-i\Gamma_{jk}^i \lambda_j^A\partial_{X_i^A}\psi+\sqrt{2}\partial_{Z^A}\partial_{\lambda_j^A}\psi=0
\\
\label{u6}
Q_{j}^{\dag}\psi&=+i\Gamma_{j8}^i \partial_{X_i^A}\partial_{\lambda_8^A}\psi+i\Gamma_{jk}^i \partial_{X_i^A}\partial_{\lambda_k^A}\psi-\sqrt{2}\partial_{\overline{Z}^A}\lambda_j^A\psi=0.
\end{align}
Here $\Gamma_{j8}^i=-i\delta_j^i$ and $\Gamma_{jk}^i= iC_{ijk}$, where the $C_{ijk}$ are the structure constants of the octonion algebra.

In order to verify \eqref{superchargecond} we only need (\ref{u3}) and (\ref{u4}). These equations are only valid at $\partial \Omega$ and they pertain the normal derivative of $\psi$ there. The $9(N^2-1)-1$ remaining angular derivatives are tangential derivatives vanishing at the boundary. That is,
\[
\partial_{\alpha_m}\psi\vert_{\partial\Omega}=0  \qquad \text{for} \quad  m=1,\dots,9(N^2-1)-1.
\]

Consider the derivative with respect to $\rho^2$. Observe that
\begin{equation}
\partial_{Z^A}\psi\vert _{\partial\Omega}=2\overline{Z}^A\frac{\partial\psi}{\partial\rho^2}\vert_{\partial\Omega}\quad \text{and} \quad
\partial_{X_i^A}\psi\vert _{\partial\Omega}=2X_i^A\frac{\partial\psi}{\partial\rho^2}\vert_{\partial\Omega}.
\end{equation}
Write $\partial_{\rho^2}\psi\equiv \psi_{\rho^2}$. Then \eqref{u3} and \eqref{u4}  reduce to
\begin{equation}\label{a}
Q_8\psi=\sqrt{2}\overline{Z}^A \partial_{\lambda_8^A}\psi_{\rho^2} + (X_i^A\lambda_i^A)\psi_{\rho^2}=0
\end{equation}
and
\begin{equation}\label{b}
Q_8^\dag\psi=-\sqrt{2}Z^A \lambda_8^A\psi_{\rho^2} -X_i^A\partial_{\lambda_i^A}\psi_{\rho^2}=0.
\end{equation}
Applying $\overline{Z}^B\frac{\partial}{\partial_{\lambda_8^B}}$ to $Q_8^\dag\psi$, gives
\begin{equation}\label{u9}
X_j^A\partial_{\lambda_j^A}(\overline{Z}^B\frac{\partial\psi_{\rho^2}}{\partial \lambda_8^B})-\sqrt{2}Z^A\overline{Z}^A\psi_{\rho^2}+ \sqrt{2}(Z^A \lambda_8^A)(\overline{Z}^B\frac{\partial\psi_{\rho^2}}{\partial \lambda_8^B})=0.\end{equation}
Replacing (\ref{a}) into (\ref{u9}), then gives
\begin{equation}\label{10}
-X_j^A \partial_{\lambda_j^A}(\frac{1}{\sqrt{2}}(X_i^B \lambda_i^B)\psi_{\rho^2})-\sqrt{2}Z^A\overline{Z}^A\psi_{\rho^2}+\sqrt{2}(Z^A\lambda_8^A)(\overline{Z}^B\frac{\partial\psi_{\rho^2}}{\partial\lambda_j^B})=0.
\end{equation}
Thus
\begin{equation}\label{11}
-(X_j^AX_j^A+2Z^A\overline{Z}^A)\psi_{\rho^2}+ (X_i^B\lambda_i^B)X_j^A\frac{\partial \psi_{\rho^2}}{\partial\lambda_j^A}+2(Z^A\lambda_8^A)(\overline{Z}^B\frac{\partial\psi_{\rho^2}}{\partial\lambda_B^A})=0.
\end{equation}

Now, on $\partial\Omega$,
\[(X_j^A)^2 +2Z^A\overline{Z}^A=R^2.
\]
Then
\[
R^2\psi_{\rho^2}\vert_{\partial\Omega}=0
\] for $R^2\ne 0$. Thus
\[
\psi_{\rho^2}\vert_{\partial\Omega}=0.
\]

By virtue of the Cauchy-Kovalevskaya Theorem, it then follows that $\psi=0$ is the unique solution in a neighbourhood of the boundary. Moreover, since the potential is analytic on $\Omega$, this solution can be extended uniquely to the whole ball $\Omega$. Thus, the supercharges $Q_{\alpha}$ and $Q_\alpha^\dag$ indeed satisfy the condition \eqref{superchargecond}.

\subsection*{The constraint}
 We now define the subspace decomposition associated to the constraint \eqref{constraint}, which is required in the framework of section~\ref{4}.  Since the $\varphi^A$ are differential operator of order 1 with $d=N^2-1$, then
\[
\varphi \equiv (\varphi^A)_1^d:C^\infty_{\mathrm{c}}(\Omega)\longrightarrow [C^\infty_{\mathrm{c}}(\Omega)]^d\subset [\cL_2(\Omega)]^d
\]
is a densely defined operator onto $d$ copies of $\cL_2(\Omega)$. Its adjoint is
\[
    \phi^\dag:\Z\longrightarrow \cL_2(\Omega)
\]
where the domain
\[
    \Z=\left\{\psi \in [\cL_2(\Omega)]^d:\begin{array}{l}
     \text{there exists }
     \lambda\in \cL_2(\Omega) \\ (\psi,\varphi \nu)_{[\cL_2(\Omega)]^d}=(\lambda,\nu)_{\cL_2(\Omega)} \\
\text{ for all }
     \nu \in C^\infty_{\mathrm{c}}(\Omega)\end{array}\right\}.
\]
Define
\[
   \X=\{\eta \in C^\infty_{\mathrm{c}}(\Omega):\varphi \eta =0\}=\ker(\varphi)
\]
and
\[
     \Y=\left\{\lambda \in C^\infty_{\mathrm{c}}(\Omega):\begin{array}{l}
     \text{for suitable }
     \psi\in \Z \\ (\lambda,\nu)_{\cL_2(\Omega)}=
     (\psi,\varphi \nu)_{[\cL_2(\Omega)]^d} \\ \text{ for all }
     \nu\in C^{\infty}_{\mathrm{c}}(\Omega)\end{array}\right\}=
     \operatorname{ran}(\varphi^{\dag}|_{[C^\infty_{\mathrm{c}}(\Omega)]^d}).
\]
Then $\X,\Y\subset C^\infty_\mathrm{c}(\Omega)$, these two spaces are orthogonal in $\cL_2(\Omega)$ and
\[
    C^\infty_\mathrm{c}(\Omega)=\X+\Y.
\]

Let $\X^1_0$ and $\Y^1_0$ be defined as the closures of $\X$ and $\Y$, respectively, in $\cH^1(\Omega)$. These two subspaces satisfy the conditions \ref{asumb}  and also the first part of the condition \ref{asumd} in the Assumption~\ref{asum}.

Since
\[
\varphi^A Q\psi= Q\varphi^A\psi \quad \text{ and }\quad \varphi^A Q^\dagger \psi= Q^\dagger \varphi^A\psi \quad \text{for all} \quad
\psi\in\X,
\]
we know that $H\eta \in \X$ for all
$\eta \in \X$. Then
\[
      H_\X\equiv H|_{\X^2\cap \X^1_0}:\X^2\cap \X^1_0\longrightarrow \X^0.
\]
Moreover, since $H$ is selfadjoint, and $\X$ and $\Y$ are orthogonal in $\cL_2(\Omega)$,
we also have $H\eta \in \Y$ for all
$\eta \in \Y$. Thus
\[
      H_\Y\equiv H|_{\Y^2\cap \Y^1_0}:\Y^2\cap \Y^1_0\longrightarrow \Y^0.
\]
This ensures the second part of the condition \ref{asumd} of the Assumption~\ref{asum}.

\subsection*{Existence and uniqueness of the ground state}
 Since \eqref{superchargecond} is fulfilled, then also \eqref{superchargecondX} is fulfilled.  The following main result is a direct consequence of  Lemma~\ref{regularity}. 
 
 Consider the boundary value problem \eqref{pro0} associated to the Hamiltonian $H$ given by \eqref{sm}, associated with the $D=11$ regularized supermembrane (the $\mathcal{N}=16$ supersymmetric $SU(N)$ matrix model). Consider
\begin{equation} \label{pro2sm}   \tag{VIII}
\begin{cases}
\begin{matrix} H\Phi= f  \\
\varphi^A\Phi =0 \end{matrix} & \textrm{in } \Omega \\
\Phi=0 & \textrm{on } \partial\Omega .
\end{cases}
\end{equation}

\begin{theorem} \label{main}
Let $\Omega$ be a $9(N^2-1)$-dimensional ball. Let $g\in \X^2$ and $f=-Hg\in \X^0$.  There always exists a unique solution $\Phi$ to the problem \eqref{pro2sm}, which lies in the space $\cH^2(\Omega)\cap \cH^1_0(\Omega)$. The corresponding solution $\Psi=\Phi+g\in \cH^2(\Omega)$ to the problem \eqref{pro0} also exists and is also unique.
\end{theorem}

%%%%%%%%%%%%%%%%%%%%%%%%%%%%%%%%%%%%%%%%%%%%%%%%%%%%%%%%%%%%%%%%%%%%%%%%%%%%%%%%%%%%%

 %%%%%%%%%%%%%%%%%%%%%%%%%%%%%%%%%%%%%%%%%%%%%%%%%%%%%%%%%%%%%%%%%%%%%%%%%%%%%%%%%%%%%
\subsection*{Invariance of the solution under $SO(9)$} \label{5.2}
 The supermembrane in 11D  in the Light Cone Gauge has a $SO(9)$ symmetry, the residual Lorentz invariance. The groundstate of the regularized supermembrane must be a singlet under $SO(9)$, in order to be related to the $D=11$ supergravity multiplet. 
 
 Denote by $J$ the generators of the algebra of $SO(9)$. Then $J$ is a first order differential operator which commutes with $H$ and satisfies $J\Phi\in \cH^1_0(\Omega)$ where $\Phi$ is the solution to Theorem~\ref{main}. Assume that $g$ is a singlet under $SO(9)$. That is $Jg=0$. Then the solution $\Phi$ to \eqref{pro2x} is also a singlet under $SO(9)$.

Note that
\begin{equation*}
(J^{\dagger}\xi, H\Phi)_{\cL_2(\Omega)}=(J^{\dagger}\xi, -Hg)_{\cL_2(\Omega)}=-(H\xi,Jg)_{\cL_2(\Omega)}=0.
\end{equation*}
Then \begin{equation*}
(H\xi,J\Phi)_{\cL_2(\Omega)}=0\qquad \text{for all } \xi\in C_{\mathrm{c}}^{\infty}(\Omega).
\end{equation*}
Hence $\fD(\xi,J\Phi)=0$ for all $\xi \in \cH^1_0(\Omega).$ We now use that $J\Phi\in \cH^1_0(\Omega)$ and the argument in the proof of Lemma~\ref{existunique},  to obtain
\begin{equation}
J\Phi=0
\end{equation}
as claimed. Since $\Psi=\Phi+g$, indeed $$J\Psi= J\Phi+Jg=0.$$

%%%%%%%%%%%%%%%%%%%%%%%%%%%%%%%%%%%%%%%%%%%%%%%%%%%%%%%%%%%%%%%%%%%%%%%%%%%%%%%%%%%%%%%%%%%%%%%%

%%%%%%%%%%%%%%%%%%%%%%%%%%%%%%%%%%%%%%%%%%%%%%%%%%%%%%%%%%%%%%%%%%%%%%%%%%%%%
\section{Conclusions}
In this paper we fully settled part \ref{parta} of the programme 
established in section~\ref{intro} for proving or disproving the existence of ground state wave functions. We showed that the ground state wavefunction for the mass operator of the regularized $D=11$ supermembrane theory\footnote{That is for the Hamiltonian of the $\mathcal{N}=16$ supersymmetric $SU(N)$ matrix model.} on a bounded smooth domain, exists and is unique. Under a suitable assumption on the boundary condition, which is a given datum, this solution corresponds to a singlet under $SO(9)$, the residual Lorentz invariance. The center of mass of the supermembrane moves freely on an 11D Minkowski spacetime but the membrane excitations are restricted to a bounded smooth domain.

Supersymmetry plays a crucial role in all the results presented in this paper. They rely on general rigorous arguments formulated in sections~\ref{3} and \ref{4}. These are valid in the context of supersymmetric theories for a Schr{\"o}dinger  Hamiltonian with a polynomial potential. The bounded domain is chosen to be invariant under the action of the symmetries of the theory. The uniqueness of the groundstate wavefunction  relies on the property (K) introduced in section~\ref{3}, and it is also satisfied by the supermembrane supersymmetric charges,

\[
    Q\psi=0 \text{ and } Q^\dag\psi=0\text{ for }\psi\in \cH^1_0(\Omega)
    \quad \Rightarrow \quad \psi=0.
\]
 The framework of sections~\ref{3} and \ref{4} provides a new approach in the context of matrix models which allows characterizing  the ground state wavefunction for a wide variety of supersymmetric matrix models by means of the homogeneous and inhomogeneous Dirichlet problems. The physical theory may  or may not possess gauge symmetry. A novel feature here is a simplification of the treatment of the gauge constraint. There is no need to solve it explicitly, as it is enough to set it as a subspace of the configuration space with natural properties arising from the gauge theory. Moreover, the analysis of section~\ref{4} represents a generalization of the analysis for constrained theories whose physical space of states is a subspace given by the kernel of an operator, as it happens in physical gauge theories.

Part~\ref{parta} of the programme described in section~\ref{intro} is a crucial step towards the solution of the ground state wavefunction problem on an unbounded domain (unbounded membrane excitations). The quest for the ground state of the regularized $D=11$ supermembrane, that is, the $\mathcal{N}=16$   supersymmetric $SU(N)$ matrix model (which is expected to be a multiplet of the $D=11$ supergravity) is a fundamental step towards the quantization of these theories, and in a more general context is fundamental in the quantization of M-theory.

The methods presented above might also have an impact in other areas of physics. These include the study of AdS/CFT black holes, compact Yang-Mills matrix models and other M-theory characterizations.

%%%%%%%%%%%%%%%%%%%%%%%%%%%%%%%%%%%%%%%%%%%%%%%%%%%%%%%%%%%%%%%%%%%%%%%%%%%%%%%%%%%%%%%%%%%%%%%%%%%%%%%%%%%%%%%%%%%%%%%%%%%%%

%%%%%%%%%%%%%%%%%%%%%%%%%%%%%%%%%%%%%%%%%%%%%%%%%%%%%%%%%%%%%%%%%%%%%%%%%%%%%%%%%%%%%%%%%%%%%%%%%%%%%%%%%%%%%%%%%%%%%%%%%%%%%
\section{Acknowledgements}
We wish to thank M.~Asorey, J.~Hoppe, D.~Lundholm, P.~Meessen and M.~Trzetrzelewski for their insightful comments.
MPGM wishes to express her sincere gratitude to Jana Bj{\"o}rn and Nages Shanmugalingam for a stimulating discussion related to $^1$. MPGM acknowledges support from the Instituto de F{\'\i}sica Te{\'o}rica IFT at the Universidad Aut{\'o}noma de Madrid and the Theoretical Physics Department at the Universidad Zaragoza (Spain) where part of the research reported in this work was conducted. MPGM acknowledges support from Mecesup grant ANT1398, Universidad de Antofagasta (Chile). AR acknowledges support from project number 1121103 Fondecyt (Chile). MPGM and AR are grateful to EU-COST Action MP1210 `The String Theory Universe'.

%%%%%%%%%%%%%%%%%%%%%%%%%%%%%%%%%%%%


\begin{thebibliography}{99}
\bibitem{manolo1}N. Acharya, M. Asorey, A.P. Balachandran and S. Vaidya, {\em Supersymmetry: Boundary Conditions and Edge States},
 arXiv:1501.00634.
%
\bibitem{manolo2} M. Asorey, D. Garcia-Alvarez, J.M. Munoz-Castaneda, {\em
Boundary Effects in Bosonic and Fermionic Field Theories},
 arXiv:1501.03752.
%
\bibitem{lin-yin} Ying-Hsuan Lin, Xi Yin, {\em On the ground state wavefunction of matrix theory}, arXiv:1402.0055.
%
\bibitem{halpern} M. Claudson, M. B. Halpern, {\em Supersymmetric Ground State wavefunctions}, Nucl.Phys. {\bf B250} (1985)  no. 4, 689.
%
\bibitem{korcyl} Piotr Korcyl, {\em Classical trajectories and quantum supersymmetry}, Phys.Rev. {\bf  D74} (2006) 115012.
%
\bibitem{kp} A.M. Khvedelidze, H.P. Pavel, {\em On the ground state of Yang-Mills quantum mechanics},
 Phys.Lett. {\bf A267} (2000) 96-100.
%
\bibitem{wosiek1} J. Wosiek, {\em Spectra of supersymmetric Yang-Mills quantum mechanics}
 Nucl.Phys. {\bf B644} (2002) 85-11.
%
\bibitem{dwhn}
B. de Wit, J. Hoppe, H. Nicolai, {\em On the quantum mechanics of supermembranes}, Nucl. Phys. {\bf B305} (1988) 545.
%
\bibitem{fh}J. Froehlich, J. Hoppe, {\em On Zero-Mass Ground States in Super-Membrane Matrix Models}, arXiv:hep-th/9701119.
%
\bibitem{nsy} N. Nakayama, K. Sugiyama, K. Yoshida, {\em Ground state of supermembrane on PP wave},
Phys.Rev. {\bf D68} (2003) 026001.
%
\bibitem{mitrze} Y. Michishita, M. Trzetrzelewski,{\em Towards the Ground State of the Supermembrane},
 Nucl.Phys.{\bf B868} (2013) 539-553.
%
\bibitem{hoppeIII} Jens Hoppe, {\em On The Construction of Zero Energy States in Supersymmetric Matrix Models III}, arXiv:hep-th/9711033.
%
\bibitem{hl}	J Hoppe, D Lundholm,	{\em On the Construction of Zero Energy States in Supersymmetric Matrix Models IV}, arXiv:0706.0353.
%
\bibitem{bhs} M. Bordemann, J. Hoppe, R. Suter, {\em Zero Energy States for SU(N): A Simple Exercise in Group Theory ?}, arXiv:hep-th/9909191.
%
\bibitem{hlt} J. Hoppe, D. Lundholm, M. Trzetrzelewski,	{\em Spin(9) Average of SU(N) Matrix Models I. Hamiltonian}, J.Math.Phys. {\bf 50} (2009) 043510.
%
\bibitem{wosiek2} J. Wosiek, {\em On the SO(9) structure of supersymmetric Yang-Mills quantum mechanics}, Phys.Lett. {\bf B619} (2005) 171-176.
%
\bibitem{hlt2} J. Hoppe, D. Lundholm, M. Trzetrzelewski, {\em Construction of the Zero-Energy State of SU(2)-Matrix Theory: Near the Origin}, Nucl.Phys. {\bf B817} (2009) 155-166.
%
\bibitem{halpern-su2} M.B. Halpern , C. Schwartz, {\em Asymptotic search for ground states of SU(2) matrix theory}, Int.J.Mod.Phys. {\bf A13} (1998) 4367-4408.
%	
\bibitem{hp} J. Hoppe, J Plefka {\em The Asymptotic ground state of SU(3) matrix theory}, hep-th/0002107.
%	
\bibitem{hoppe} J Hoppe, {\em Asymptotic zero energy states for SU(N) greater than or equal to 3}, hep-th/9912163.
%	
\bibitem{haho} D. Hasler, J. Hoppe {\em Asymptotic factorization of the ground state for SU(N) invariant supersymmetric matrix models}, hep-th/0206043.
%	
\bibitem{fghhy} J. Frohlich, G.M. Graf, D. Hasler, J. Hoppe, S-T. Yau, {\em Asymptotic form of zero energy wavefunctions in supersymmetric matrix models}, Nucl.Phys. {\bf B567} (2000) 231-248.
%
\bibitem{bst} E. A. Bergshoeff, E. Sezgin, and P. K. Townsend, {\em Supermembranes and eleven-dimensional supergravity}, Phys. Lett.{\bf B189} (1987), no. 1-2, 75.
%
\bibitem{goldstone} J. Goldstone, unpublished.
%
\bibitem{hoppe-phd} J. Hoppe, Quantum theory of a massless relativistic surface, Ph.D. thesis, MIT,
1982.
%
\bibitem{dwmn}
B. de Wit, U. Marquard, H. Nicolai, {\em Area preserving diffeomorphisms and supermembrane lorentz invariance}, Commun. Math. Phys. {\bf 128} (1990) 39-62.
%

\bibitem{flume} R. Flume, {\em On Quantum Mechanics With Extended Supersymmetry and Nonabelian Gauge Constraints}, Ann. of Phys. {\bf 164} (1985) 189.
%
\bibitem{baake}M. Baake, M. Reinicke and V. Rittenberg,  {\em Fierz Identities for Real Clifford Algebras and the Number of Supercharges}
J. Math.{\bf 26} (1985) 1070.
%
\bibitem{bfss}T. Banks, W. Fischler, S.H. Shenker and L. Susskind, {\em M theory as a matrix model: A Conjecture}
 Phys.Rev. {\bf D55} (1997) 5112-5128.
%
\bibitem{dwln}
B. de Wit, M. Luscher, H. Nicolai, {\em The supermembrane is unstable}, Nucl. Phys. {\bf B320} (1989) 135.
%
\bibitem{dwpp}
B. de Wit, K. Peeters, J. Plefka, {\em Supermembranes with winding}, Phys. Lett. {\bf B409} (1997) 117-123.
%
\bibitem{mrt}
I. Martin, A. Restuccia, R. S. Torrealba, {\em On the stability of compactified D = 11 supermembranes}, Nucl. Phys. {\bf B521} (1998) 117-128.
%
\bibitem{gmr} M.P. Garcia del Moral, A. Restuccia, {\em On the spectrum of a noncommutative formulation of the $D=11$ supermembrane with winding}, Phys. Rev. {\bf D66} (2002) 045023.
%
\bibitem{bgmmr}
L. Boulton, M. P. Garcia del Moral, I. Martin, A. Restuccia, {\em On the spectrum of a matrix model for the $D=11$ supermembrane compactified on a torus with non-trivial winding}, Class. Quant. Grav. {\bf 19} (2002) 2951.
%
\bibitem{bgmr}
L. Boulton, M.P. Garcia del Moral, A. Restuccia, {\em Discreteness of the spectrum of the compactified $D=11$ supermembrane with non-trivial winding}, Nucl. Phys. {\bf B671} (2003) 343-358.
%
\bibitem{br}
L. Boulton and A. Restuccia, {\em The Heat kernel of the compactified $D=11$ supermembrane with non-trivial winding},
 Nucl. Phys. {\bf B724} (2005) 380-396.%
%
\bibitem{bgmr2}
L. Boulton, M.P. Garcia del Moral, A. Restuccia, {\em The supermembrane with central charges: (2+1)-D NCSYM, confinement and phase transition}, Nucl.Phys. {\bf B795} (2008) 27-51.
%
\bibitem{bgmr3} L. Boulton, M.P. Garcia del Moral, A. Restuccia,
{\em Spectral properties in supersymmetric matrix models}, Nucl.Phys. {\bf B856} (2012) 716-747.
%
\bibitem{bmn}D. Berenstein, J. Maldacena, H. Nastase, {\em Strings in flat space and pp waves from N=4 superYang-Mills},
  JHEP{\bf  0204} (2002) 013.
	%
\bibitem{octonions} L. Boulton, M.P. Garcia del Moral, A. Restuccia, {\em On the ground state octonionic matrix models in a ball},
 Phys. Lett. B {\bf 744} (2015) 260-262.
%
\bibitem{su2} L. Boulton, M.P. Garcia del Moral, A. Restuccia, {\em Massless ground state for a compact $SU(2)$ matrix model in 4D}, arXiv:1503.05462.
%
\bibitem{Ghh} G.M. Graf, D. Hasler, J. Hoppe, {\em No zero energy states for the supersymmetric $x^2y^2$ potential}, arXiv:math-ph/0109032.
%
\bibitem{Kato} T.~Kato, Perturbation Theory of Linear Operators. Springer Verlag, Berlin, 1995.
%
\bibitem{Folland} G.~Folland, Introduction to Partial Differential Equations. Second Edition. Princeton University Press, Princeton, 1995.
%
\bibitem{1997Evans} L.~Evans, Partial Differential Equations. Second Edition. American Mathematical Society, Providence, 2010.

\vspace{2cm}

{\footnotesize
\noindent{\scshape Lyonell Boulton\\
Department of Mathematics \& Maxwell Institute for the Mathematical Sciences\\
Heriot-Watt University\\
Edinburgh, EH14 4AS, United Kingdom}\\
{\itshape E-mail address}: \href{l.boulton@hw.ac.uk}{l.boulton@hw.ac.uk}\\
{\itshape URL}: \url{http://www.ma.hw.ac.uk/~lyonell/}

\bigskip
\bigskip


\noindent{\scshape Mar\'ia Pilar Garc\'ia del Moral\\
Departamento de F\'isica\\
Universidad de Antofagasta\\
Avda Universidad de Antofagasta 02800, Antofagasta, Chile.}\\
{\itshape E-mail address}: \href{maria.garciadelmoral@uantof.cl}{maria.garciadelmoral@uantof.cl}\\
{\itshape URL}: \url{http://faciba.uantof.cl/index.php/}


\bigskip
\bigskip

\noindent{\scshape Alvaro Restuccia\\
Departamento de F\'isica\\
Universidad de Antofagasta \\
Avda Universidad de Antofagasta 02800, Antofagasta, Chile.;\\
Departamento de F\'isica, Universidad Sim\'on Bol\'ivar,\\ Valle
de Sartenejas,
1080-A Caracas, Venezuela.}\\
{\itshape E-mail address}: \href{arestu@usb.ve}{arestu@usb.ve}\\
{\itshape URL}: \url{http://faciba.uantof.cl/index.php/}
}
\clearpage


\end{thebibliography}
\end{document}